\documentclass[a4paper,11pt,twoside,notitlepage,reqno]{amsart}

\usepackage{amsmath,amssymb,amsbsy,amsthm}
\usepackage[hmargin=1.4in,vmargin=1.4in]{geometry}
\usepackage{url}

\usepackage{tikz}
\usetikzlibrary{matrix}
\usepackage[linktocpage]{hyperref}
\hypersetup{
    colorlinks=true,        
    linkcolor=red,          
    citecolor=red,         
    filecolor=magenta,      
    urlcolor=cyan           
}

\addtocontents{toc}{\protect\setcounter{tocdepth}{1}}

\usepackage{eucal}

\usepackage{algorithm}
\usepackage{algpseudocode}

\newcommand{\N}{ \mathbb{N} }

\theoremstyle{plain}
	\newtheorem{thm}{Theorem}
	\newtheorem{quest}{Question}
	
		\numberwithin{thm}{section}
	\newtheorem{lemma}[thm]{Lemma}
	\newtheorem{prop}[thm]{Proposition}

	\newtheorem*{thm*}{Theorem}
	\newtheorem*{lemma*}{Lemma}
	\newtheorem*{prop*}{Proposition}
	\newtheorem*{cor*}{Corollary}
	\newtheorem*{conj*}{Conjecture}

\theoremstyle{definition}
	\newtheorem{example}[thm]{Example}
	\newtheorem*{example*}{Example}
	
	\newtheorem{remark}[thm]{Remark}

\begin{document}


\title[Dynamical sequences]{Dynamical sequences: closure properties and automatic identity proving}

\author{Jason P. Bell}
\address{Department of Pure Mathematics\\
University of Waterloo\\
Waterloo, ON N2L 3G1\\
Canada}
\email{jpbell@uwaterloo.ca}

\author{Yuxuan Sun}
\address{Department of Mathematics\\
University of Toronto\\
Toronto, Ontario M5S 2E4\\
Canada}
\email{austin.sun@mail.utoronto.ca}
\begin{abstract}
Given an algebraically closed field $K$, a \emph{dynamical sequence} over $K$ is a $K$-valued sequence of the form
$a(n):= f(\varphi^n(x_0))$, where $\varphi\colon X\dasharrow X$ and $f\colon X\dasharrow\mathbb{A}^1$ are rational maps
defined over $K$, and $x_0\in X$ is a point whose forward orbit avoids the indeterminacy loci
of $\varphi$ and $f$.  Many classical sequences from number theory and algebraic combinatorics fall under this dynamical
framework, and we show that the class of dynamical sequences enjoys numerous closure properties and encompasses all elliptic divisibility sequences, all Somos sequences, and all $C^n$- and $D^n$-finite sequences for all $n\ge 1$, as defined by Jiménez-Pastor, Nuspl, and Pillwein. We also give an algorithm for proving that two dynamical sequences are identical and illustrate how to use this algorithm by showing how to prove several classical combinatorial identities via this method.
\end{abstract}
\maketitle

\tableofcontents

\section{Introduction}

A \textit{rational dynamical system} is a pair $(X,\varphi)$, where $X$ is a quasiprojective variety defined over a field $K$, and $\varphi:X\dashrightarrow X$ is a rational map. The forward $\varphi$-orbit of a point $x_0\in X$ is given by
\[ O_\varphi(x_0):= \{x_0,\varphi(x_0),\varphi^2(x_0),\ldots\} \]
 as long as this orbit is defined (\textit{i.e.}, $x_0$ is outside the indeterminacy locus of $\varphi^n$ for every $n\geq 0$).  Here $\varphi^n$ denotes the $n$-fold composition 
 $\varphi\circ \varphi \circ \cdots \circ \varphi$.  A \emph{dynamical sequence} over $K$ is a sequence of the form $ f\circ \varphi^n(x_0)$, where $(X,\varphi)$ is a rational dynamical system, $f:X\dashrightarrow \mathbb{A}^1$ is a rational map, and $x_0\in X$, where we assume that the orbit of $x_0$ is outside the indeterminacy locus of $\varphi^n$ for every $n$ and $\varphi^n(x_0)$ is outside of the indeterminacy locus of $f$ for every $n$.  We then say that $(X,\varphi, x_0, f)$ is the \emph{geometric data} that generates the sequence $a(n)=f\circ \varphi^n(x_0)$.  For example, if $X=\mathbb{A}^2$, $\varphi: \mathbb{A}^2\to \mathbb{A}^2$ is the map $\varphi(x,y)=(x+1,yx)$, $x_0=(1,1)$, and $f(x,y)=y$, then it is straightforward to check that $\varphi^n(1,1)=(n+1,n!)$ for all $n\ge 0$ (where $\varphi^0$ is taken to be the identity map) and so the sequence associated to the data $(X,\varphi, x_0,f)$ is the sequence $a(n)=n!$. This example illustrates how classical sequences such as the factorials naturally arise from simple algebraic dynamical systems.

In several earlier works \cite{BGS, BHS, BCH}, the authors studied the class of sequences generated in this way and showed that many
classical sequences from number theory and algebraic combinatorics fall under this framework.  In particular, the class of dynamical sequences includes all sequences whose generating functions are $D$-finite; that is, sequences whose generating functions satisfy homogeneous linear differential
equations with rational function coefficients.  This is an important class of power series since it
appears ubiquitously in algebra, combinatorics, and number theory. In particular, this class contains:
\begin{itemize}
\item all hypergeometric series (see, for example, \cite{G09, WZ});
\item generating functions for many classes of lattice walks \cite{DHRS18};
\item diagonals of multivariate rational functions \cite{Lipshitz88};
\item power series expansions of algebraic functions \cite[Chapter 6]{Stan};
\item generating series for the cogrowth of many finitely presented groups \cite{GP};
\end{itemize}
As it turns out, the class of dynamical sequences is significantly broader than the class of $D$-finite sequences.  In this paper, we show that the class of dynamical sequences contains the following classes:
\begin{itemize}
\item all $C^n$- and $D^n$-finite sequences (as defined by  Jiménez-Pastor, Nuspl, and Pillwein \cite{C2}) for every $n\ge 1$ (see Proposition \ref{prop:Cn});
\item Somos sequences that have finitely many zero values (see Proposition \ref{prop:somos});
\item Elliptic divisibility sequences that have finitely many zero values (see Proposition \ref{prop:EDS});
\item all sequences of the form $f(d^n)$ where $d$ is a nonnegative integer and $f:\mathbb{N}\to K$ is a sequence satisfying a linear recurrence over the field $K$ (see Proposition \ref{prop:LRS}(i));
\item all sequences of the form $f(P(n))$ where $P(x)\in \mathbb{Q}[x]$ is a polynomial that maps $\mathbb{N}$ to $\mathbb{N}$ (see Proposition \ref{prop:LRS}(ii)). 
\end{itemize}
Definitions for these classes of sequences appear in Section \ref{sec:exam}.

We are also able to prove that dynamical sequences enjoy numerous closure properties.
\begin{thm}
\label{thm:main}
Let $K$ be an algebraically closed field.  Then the following hold:
\begin{enumerate}
\item (closure under sum) if $a(n), b(n)$ are dynamical sequences over $K$ then so is $a(n)+ b(n)$;
\item (closure under product) if $a(n), b(n)$ are dynamical sequences over $K$ then so is $a(n)\cdot b(n)$; 
\item (closure under taking partial sums) if $a(n)$ is a dynamical sequence over $K$ then so is the sequence $a(0)+\cdots +a(n)$;
\item (closure under taking partial products) if $a(n)$ is a dynamical sequence over $K$ then so is the sequence $a(0)\cdots a(n)$;
\item (closure under shifts) if $a(n)$ is a dynamical sequence over $K$ then so is the sequence $\{a(n+i)\}_{n\ge 0}$ for each $i\ge 0$;
\item (closure under modification of initial values) if $a(n)$ is a dynamical sequence over $K$ and if there is some $N$ and some $j\le N$ such that $b(n)=a(n-j)$ for all $n\ge N$, then $b(n)$ is also a dynamical sequence;
\item (closure under subsequences along arithmetic progressions) if $d\ge 1$ and $a(n)$ is a dynamical sequence over $K$ then so is the sequence $a(dn+i)$ for $i=0,\ldots ,d-1$;
\item (closure under flooring of indices) if $d\ge 1$ and $a(n)$ is a dynamical sequence over $K$ then so is the sequence $a(\lfloor n/d\rfloor)$;
\item (closure under interlacing) if $a_0(n),\ldots ,a_{s-1}(n)$ are dynamical sequences over $K$ then so is the sequence $c(sn+i):= a_i(n)$ for $i=0,\ldots ,s-1$ and $n\ge 0$.
\end{enumerate}
\end{thm}
Since constant sequences are dynamical sequences, we then see that the set of dynamical sequences over a field $K$ forms a $K$-algebra under pointwise operations. Notably, we are unable to prove that dynamical sequences are closed under convolution.  We leave this as an open question (see Question \ref{quest:1}).

Finally, we note that Wilf and Zeilberger \cite{WZ} considered the question of when two sequences $a(n)$ and $b(n)$ whose generating functions are $D$-finite have the property that $a(n)=b(n)$ for all $n\ge 0$.  As it turns out, many interesting combinatorial identities can be phrased in this framework and Wilf and Zeilberger show that in many cases these identities can be verified ``automatically'' by checking that the two sequences agree up to a certain bound.  We consider the analogous question for dynamical sequences and show that there is an algorithm for determining whether $a(n)=b(n)$ for all $n\ge 0$ for two $K$-valued dynamical sequences. Specifically, we prove the following.
\begin{thm} Given two dynamical sequences $a(n)$ and $b(n)$ over a field $K$, there is an algorithm which terminates after a finite number of steps and determines whether the identity $a(n)=b(n)$ holds for all $n \in \mathbb{N}$.
\label{thm:alg}
\end{thm}
In fact, one can give explicit bounds on the number of steps required for the algorithm to terminate in terms of the geometric data associated to the sequences $a(n)$ and $b(n)$ (see Remark \ref{rem:bounds}), but in practice the algorithm terminates much more quickly than these worst-case bounds and we give several examples illustrating how to employ this algorithm to prove equality of two dynamical sequences.

\subsection{Conventions} Throughout, $\N:= \{0,1,2,3,\ldots\}$ and $K$ will be an algebraically closed field.

\subsection{Organization}
In Section~\ref{sec:closure}, we prove the closure properties described in Theorem~\ref{thm:main}. Section~\ref{sec:exam} explores examples of dynamical sequences, showing that the class includes several classical sequence families. In Section~\ref{sec:alg}, we establish Theorem~\ref{thm:alg} and demonstrate the identity-proving algorithm in several concrete cases. We conclude by posing some open questions about dynamical sequences in Section~\ref{sec:questions}.

\subsection*{Acknowledgments} J. P. Bell was supported by NSERC grant RGPIN-2022-02951. Y. Sun was partially supported by a NSERC CGS M Award for his
Master of Mathematics thesis, which led to the completion of parts of this project.
\smallskip

\section{Closure properties} \label{sec:closure}
In this section, we prove Theorem~\ref{thm:main}.

\begin{proof}[Proof of Theorem~\ref{thm:main}]
We begin by proving closure under sum and product. Let $a(n)$ and $b(n)$ be dynamical sequences over $K$ generated by geometric data $(X,\phi,x_0,f)$ and $(Y,\psi,y_0,g)$, respectively; that is, $a(n) = f(\phi^n(x_0))$ and $b(n) = g(\psi^n(y_0))$ for all $n \geq 0$. Consider the rational map $\phi \times \psi \colon X \times Y \dashrightarrow X \times Y$, and define rational maps $h_1, h_2 \colon X \times Y \dashrightarrow \mathbb{A}^1$ by
\[
h_1(x,y) = f(x)g(y) \quad \text{and} \quad h_2(x,y) = f(x) + g(y).
\]
Then $h_1((\phi \times \psi)^n(x_0, y_0)) = a(n)b(n)$ and $h_2((\phi \times \psi)^n(x_0, y_0)) = a(n) + b(n)$, establishing items (1) and (2).

To prove closure under partial sums and products, suppose $(X,\phi,x_0,f)$ generates $a(n)$. Let $Z = X \times \mathbb{A}^1$, and define a rational map $\chi \colon Z \to Z$ by $\chi(x,p) = (\phi(x), p + f(x))$. Let $z_0 = (x_0, 0)$. Then by induction, 
\[
\chi^n(z_0) = (\phi^n(x_0), a(0) + \cdots + a(n-1)).
\]
Define $h \colon Z \to \mathbb{A}^1$ by $h(x,p) = f(x) + p$. Then
\[
h(\chi^n(z_0)) = a(0) + \cdots + a(n),
\]
proving (3). Closure under partial products is proved the same way, but where we take $\chi(x,p) = (\phi(x), p\cdot f(x))$, $z_0=(x_0,1)$, and $h(x,p)=f(x)\cdot p$. 

Closure under forward shifts is immediate: if $(X,\phi,x_0,f)$ generates $a(n)$, then $(X,\phi,\phi^i(x_0),f)$ generates $a(n+i)$, verifying (5).

To prove closure under modification of initial values, suppose that $j\le N$ and $b(n)=a(n-j)$ for all $n \geq N$, and that $(X,\phi,x_0,f)$ generates $a(n)$. Consider the disjoint union $W := \bigsqcup_{i=0}^N X_i$, where each $X_i$ is a copy of $X$. Define a rational map $\Psi \colon W \to W$ by
\[
\Psi|_{X_i}(x) = 
\begin{cases}
x \in X_{i+1} & \text{if } 0 \leq i < N,\\
\phi(x) \in X_N & \text{if } i = N.
\end{cases}
\]
Define a rational function $g \colon W \dashrightarrow \mathbb{A}^1$ by setting $g|_{X_i}$ to be the constant function $b(i)$ for $i < N$ and $g|_{X_N} = f$. Then the geometric data $(W, \Psi, \phi^{N-j}(x_0) , g)$ generates $b(n)$, proving (6).

We can prove (7) quickly: if $a(n)$ is generated by the data $(X,\phi, x_0, f)$ then $(X,\phi^d, \phi^i(x_0), f)$ generates $a(dn+i)$.

To prove closure under flooring of indices, let $d \geq 1$ and suppose $(X,\phi,x_0,f)$ generates $a(n)$. Let $U = X^d$ and define $u_0 = (x_0, \ldots, x_0) \in U$. Define a rational map $\mu \colon U \to U$ by
\[
\mu(u_1, \ldots, u_d) = (\phi(u_d), u_1, \ldots, u_{d-1}).
\]
By induction, for all $n \geq 0$ and $0 \leq i < d$, we have
\[
\mu^{dn+i}(u_0) = (\phi^{n+1}(x_0), \ldots, \phi^{n+1}(x_0), \phi^n(x_0), \ldots, \phi^n(x_0)),
\]
with first $i$ coordinates equal to $\phi^{n+1}(x_0)$. Define $e \colon U \to \mathbb{A}^1$ by $e(u_1,\ldots,u_d) = f(u_d)$. Then $e(\mu^n(u_0)) = a(\lfloor n/d \rfloor)$, proving (8).

To prove closure under interlacings, we follow the method of \cite[Theorem 5.1(6)]{C2}. Let $a_0(n),\ldots,a_{s-1}(n)$ be dynamical sequences. For each $i$, the sequence $a_i(\lfloor n/s \rfloor)$ is dynamical by (8). Define
\[
c_i(n) = 
\begin{cases}
1 & \text{if } n \equiv i \pmod{s},\\
0 & \text{otherwise}.
\end{cases}
\]
Each $c_i(n)$ satisfies a linear recurrence and hence is a dynamical sequence. Then
\[
b(n) = \sum_{i=0}^{s-1} c_i(n) \cdot a_i\left( \left\lfloor \frac{n}{s} \right\rfloor \right)
\]
is a dynamical sequence by closure under sums and products, establishing (9).
\end{proof}

\section{Some notable classes of dynamical sequences}\label{sec:exam}

In this section, we give a brief overview of several notable families of sequences from combinatorics and number theory that lie within the class of dynamical sequences. Our main result is the following theorem.

\begin{thm}\label{thm:main2}
The following sequences are dynamical sequences:
\begin{enumerate}
    \item $D$-finite sequences and more generally $D^n$-finite sequences for every $n$;
    \item Somos sequences;
    \item Elliptic divisibility sequences;
    \item Sequences of the form $f(P(n))$, where $f:\mathbb{N}\to K$ is a linear recurrence sequence and $P(x)\in \mathbb{Z}[x]$ takes nonnegative values on $\mathbb{N}$;
    \item Sequences of the form $f(a^n)$, where $f:\mathbb{N}\to K$ is a linear recurrence sequence and $a\ge 1$ is a positive integer.
\end{enumerate}
\end{thm}

Definitions and constructions for the first three classes appear in the subsections below. The proof of Theorem~\ref{thm:main2} relies on the following proposition, which allows one to construct new dynamical sequences from previously known ones.

\begin{prop}\label{prop:R}
Let $K$ be a field and let $c_1(n),\ldots,c_s(n)$ be $K$-valued dynamical sequences. Suppose that $f:\mathbb{N}\to K$ satisfies a recurrence of the form
\[
f(n) = R(c_1(n), \ldots, c_s(n), f(n-1), \ldots, f(n-d))
\]
for all $n \ge d$, where $R(y_1,\ldots,y_s,x_1,\ldots,x_d)\in K(y_1,\ldots,y_s,x_1,\ldots,x_d)$ is a rational function regular at the evaluation point for each such $n$. Then $f(n)$ is a dynamical sequence.
\end{prop}

\begin{proof}
Let $(X_i, \phi_i, x_0^{(i)},f_i)$ for $i=1,\ldots,s$ be geometric data for the sequences $c_1(n),\ldots,c_s(n)$. Define the variety
\[
Y := X_1 \times \cdots \times X_s \times \mathbb{A}^d,
\]
and define a rational map $\Psi : Y \to Y$ by
\begin{align*}
&~ \Psi(p_1,\ldots,p_s, a_1,\ldots,a_d) \\
&= \left( \phi_1(p_1), \ldots, \phi_s(p_s), a_2,\ldots,a_d, R(f_1(p_1),\ldots,f_s(p_s), a_1,\ldots,a_d) \right).
\end{align*}
Let $y_0 = (x_0^{(1)},\ldots,x_0^{(s)}, f(0),\ldots,f(d-1))$. By induction, we find
\[
\Psi^n(y_0) = (\phi_1^n(x_0^{(1)}), \ldots, \phi_s^n(x_0^{(s)}), f(n), \ldots, f(n+d-1)).
\]
Defining $g: Y \to \mathbb{A}^1$ as the projection onto the first coordinate of the $\mathbb{A}^d$ factor, we obtain $f(n) = g(\Psi^n(y_0))$, proving that $f(n)$ is a dynamical sequence.
\end{proof}

We now apply this result to several concrete examples of sequences of interest.

\subsection{$D^n$-finite sequences}

The concept of $D$-finite sequences (i.e., sequences satisfying linear recurrences with polynomial coefficients) is central within enumerative combinatorics and symbolic computation. The notion has been extended to $D^n$-finite sequences by Jim\'enez-Pastor, Nuspl, and Pillwein \cite{C2}, where for $n\ge 2$ a $D^n$-finite sequence satisfies a recurrence with $D^{n-1}$-finite coefficients, and $D^1$-finite sequences are those that are holonomic (or, equivalently, whose generating series are $D$-finite). Analogously, the authors also define $C^n$-finite sequences, where the recurrence coefficients are $C^{n-1}$-finite, with $C^1$ denoting linear recurrence sequences. 

Thus, $C^n$-finite sequences are contained in the class of $D^n$-finite sequences, and $D^n$-finite sequences are contained in the $C^{n+1}$-finite ones. These hierarchies capture sequences beyond the holonomic class. For instance, $a(n) = F_{n^2}$, where $F_n$ is the $n$-th Fibonacci number, is $D^2$-finite but not $D$-finite \cite[Example 2.2]{C2}.

\begin{prop}\label{prop:Cn}
The class of dynamical sequences contains all $C^n$- and $D^n$-finite sequences for every $n \ge 1$.
\end{prop}

\begin{proof}
It suffices to prove this for $D^n$-finite sequences. We proceed by induction on $n$. For $n = 1$, $D^1$-finite sequences are holonomic, and these are dynamical (see \cite[\S 6]{BCH}). Assume the result for $n = m$ and let $f(n)$ be $D^{m+1}$-finite. Then there exists $s \ge 1$ and $D^m$-finite sequences $a_0(n), \ldots, a_s(n)$ such that
\[ a_0(n) f(n) = a_1(n) f(n-1) + \cdots + a_s(n) f(n-s), \quad \text{for all } n \ge s, \]
with $a_0(n) \ne 0$ on this domain. Define
\[ R(x_0, x_1, \ldots, x_s, y_1, \ldots, y_s) = \frac{x_1 y_1 + \cdots + x_s y_s}{x_0}, \]
and apply Proposition~\ref{prop:R} to conclude that $f(n)$ is a dynamical sequence.
\end{proof}

\subsection{Somos sequences}

Somos sequences were introduced by Somos (see Gale \cite{Gale}). These are sequences $a(n)$ satisfying a bilinear recurrence of the form
\begin{equation}\label{eq:Somos}
a(n) a(n-k) = \sum_{i=1}^{\lfloor k/2 \rfloor} a(n-i) a(n-k+i),
\end{equation}
for some fixed $k$ and all $n \ge k$. To compute $a(n)$ from this recurrence, one needs $a(i) \ne 0$ for all relevant indices.

\begin{prop}\label{prop:somos}
Let $a(n)$ be a Somos sequence taking the value zero only finitely many times. Then $a(n)$ is a dynamical sequence.
\end{prop}

\begin{proof}
Let $N$ be large enough that $a(n) \ne 0$ for all $n \ge N$. Let $k$ be the recurrence order. Define $X = \mathbb{A}^k$ and set
\[ \phi(x_1, \ldots, x_k) = \left( x_2, \ldots, x_k, \left( \sum_{i=1}^{\lfloor k/2 \rfloor} x_{k+1-i} x_{1+i} \right) x_1^{-1} \right). \]
Then the recurrence gives
\[ \phi(a(n-k), \ldots, a(n-1)) = (a(n-k+1), \ldots, a(n)) \quad \text{for all } n \ge N + k. \]
Let $p_0 = (a(N), \ldots, a(N+k-1))$ and let $f$ be the projection to the last coordinate. Then $a(n+N) = f(\phi^n(p_0))$, and closure under shift (Theorem~\ref{thm:main}(5)) implies $a(n)$ is a dynamical sequence.
\end{proof}

\subsection{Elliptic divisibility sequences}

Elliptic divisibility sequences (EDS) arise from the study of points on elliptic curves and generalize classical divisibility sequences such as the Fibonacci numbers. See \cite{Everest} and \cite{Ward} for comprehensive accounts. An EDS satisfies a recurrence of the form:
\[ W_{n+2}W_{n-2}W_1^2 = W_{n+1}W_{n-1}W_2^2 - W_3W_1W_n^2, \quad \text{for all } n > 2. \]

\begin{prop}\label{prop:EDS}
Let $a(n)$ be an elliptic divisibility sequence. If $a(n)$ has finitely many zero terms, then $a(n)$ is a dynamical sequence.
\end{prop}

\begin{proof}
This proof follows the same strategy as the proof of Proposition~\ref{prop:somos}. The recurrence can be encoded in a rational map provided $a(n) \ne 0$ for all but finitely many $n$.
\end{proof}
\begin{remark} We note that a deep result due to Mazur \cite{Maz} about rational torsion points of elliptic curves gives as an application that a rational EDS $a(n)$ with $a(i)$ nonzero for $i\in \{2,3,\ldots ,12\}\setminus \{11\}$ then $a(n)\neq 0$ for $n\ge 1$ (see Gezer \cite[Sec. 3]{Gez}). \label{rem:Mazur}
\end{remark}

\subsection{Subsequences of linear recurrence sequences}
We now consider certain subsequences of a sequence $f(n)$ that satisfies a linear recurrence over $K$.
\begin{lemma} Let $\lambda\in K$.  Then the following hold:
\begin{enumerate}
\item[(i)] if $d$ is a positive integer then $a(n)=\lambda^{d^n}$ is a dynamical sequence over $K$;
\item[(ii)] if $P(x)\in \mathbb{Q}[x]$ is a polynomial with the property that $P(\mathbb{N})\subseteq \mathbb{N}$ then $b(n)=\lambda^{P(n)}$ is a dynamical sequence over $K$.
\end{enumerate}
\label{lem:lambda}
\end{lemma}
\begin{proof}
To prove (i), we let $\psi: \mathbb{A}^1\to \mathbb{A}^1$ be the map $\psi(x)=x^d$ then $\psi^m(\lambda)=\lambda^{d^m}$ and so if we post compose with the identity map, we see that $\lambda^{d^m}$ is a dynamical sequence as well.  

To prove (ii), we note that a polynomial that takes nonnegative integer values at natural numbers is expressible in the form
$P(x) = \sum_{i=0}^d c_i {x\choose i}$ where $c_0,\ldots ,c_d$ are nonnegative integers by a classical result on binomial coefficient expansions of integer-valued polynomials (see Pólya and Szegő \cite[Problem 85, p. 129]{PS}). 
We now let $\lambda\in \bar{K}$ and we construct a dynamical sequence using the map
$\phi: \mathbb{A}^d\to \mathbb{A}^{d}$ as follows
$\phi(a_1,\ldots, a_d)= (\lambda a_1, a_2a_1,\ldots, a_d a_{d-1})$.  
Starting with $x_0=(1,1,\ldots ,1)$, we claim that 
$$\phi^n(x_0) = \left(\lambda^n, \lambda^{{n\choose 2}},\ldots ,\lambda^{{n\choose d}}\right).$$ We prove this by induction on $n$, with the base case when $n=0$ being immediate.  
Then observe that if the claim holds when $n=m$, then
\begin{align*}
\phi^{m+1}(x_0) &= \phi(\phi^m(x_0)) \\
&= \phi\left( \lambda^m, \lambda^{{m\choose 2}},\ldots ,\lambda^{{m\choose d}} \right) \\
&= \left( \lambda^{m+1}, \lambda^{{m\choose 2}+m},\ldots ,\lambda^{{m\choose d} + {m\choose d-1}} \right) \\
&= \left( \lambda^{m+1}, \lambda^{{m+1\choose 2}},\ldots ,\lambda^{{m+1\choose d}} \right).
\end{align*}
We now let $f:\mathbb{A}^d\to \mathbb{A}^1$ be the map $f(u_1,\ldots ,u_d) = \lambda^{c_0} u_1^{c_1}\cdots u_d^{c_d}$.  Then 
$$f(\phi^n(x_0))=\lambda^{P(n)}.$$  
\end{proof}
Using the preceding lemma, we can quickly prove the following result.
\begin{prop} Let $f(n)$ be a $K$-valued sequence that satisfies a linear recurrence over $K$.  Then the following hold:
\begin{enumerate}
\item[(i)] if $d$ is a nonnegative integer then $a(n)=f(d^n)$ is a dynamical sequence over $K$;
\item[(ii)] if $P(x)\in \mathbb{Q}[x]$ is a polynomial with the property that $P(\mathbb{N})\subseteq \mathbb{N}$ then $b(n)=f(P(n))$ is a dynamical sequence over $K$.
\end{enumerate}
\label{prop:LRS}
\end{prop}
\begin{proof}
For (i), by \cite[Sec. 1.1.6]{Everest} there exists a natural number $M$ and there exist $\lambda_1,\ldots ,\lambda_s\in K$, and constant $c_{i,j}\in K$ with $1\le i\le s$, $0\le j\le M$, such that 
$$f(n) = \sum_{i,j} c_{i,j} {n\choose j} \lambda_i^n.$$
Now $f(d^n) = \sum_{i,j} c_{i,j} {d^n \choose j} \lambda_i^{d^n}$.  By Lemma \ref{lem:lambda}, $\lambda_i^{d^n}$ is a dynamical sequence over $K$ for $i=1,\ldots ,s$.  
If $K$ has characteristic zero, then for fixed $j$ the sequence ${d^n\choose j}= d^n(d^n-1)\cdots (d^n-j+1)/j!$ (as a sequence in $n$) is a finite product of dynamical sequences over $K$ and hence is itself a dynamical sequence by Theorem \ref{thm:main}(2).  Thus using closure under sum and product of dynamical sequences (Theorem \ref{thm:main}(1)--(2)), we have $f(d^n)$ is a dynamical sequence.  Similarly, if $K$ has characteristic $p>0$ then one can easily see that for fixed $j$, ${d^n\choose j}$ is eventually periodic mod $p$ and hence satisfies a linear recurrence and so it is a dynamical sequence. Thus ${d^n\choose j}$ is a dynamical sequence over $K$ for $0\le j\le M$ and so we again use closure under sums and products to obtain the result for $f(d^n)$. 

For (ii), we argue similarly and have 
$$f(P(n)) = \sum_{i,j} c_{i,j} {P(n)\choose j}\cdot \lambda_i^{P(n)},$$ and we again use Lemma \ref{lem:lambda}.
\end{proof}
\vskip 2mm
\subsection{Proof of Theorem \ref{thm:main2}} 
We now observe that Theorem \ref{thm:main2} follows by combining Propositions \ref{prop:Cn}, \ref{prop:somos}, \ref{prop:EDS}, and \ref{prop:LRS}.
\section{Automatic identity checking}
\label{sec:alg}
The Wilf-Zeilberger (WZ) method \cite{WZ} is a powerful and algorithmic tool for proving identities involving $D$-finite sequences (that is, sequences satisfying linear recurrences with polynomial coefficients). This technique has found particular success in hypergeometric summation and has significantly expanded the range of identities that can be rigorously established without resorting to \emph{ad hoc} manipulations.  We show here that one can similarly ``automatically'' verify when two dynamical sequences $a(n)$ and $b(n)$ agree for all $n$ and illustrate this with some key examples.

\subsection{Algorithm}
We describe an algorithm that takes as input two dynamical sequences $a(n)$ and $b(n)$ over a common algebraically closed field $K$ and determines whether $a(n)=b(n)$ for all $n$. Notably, we will show that this algorithm always terminates and either verifies that $a(n)=b(n)$ for all $n$ or produces some $n_0$ for which $a(n_0)\neq b(n_0)$.

To describe the algorithm, we assume that we have geometric data $(X,\phi, x_0, f)$ and $(Y,\psi, y_0,g)$ that generate dynamical sequences $a(n)$ and $b(n)$ respectively.  We then let $(Z,\chi, z_0,h)$ be the geometric data given by $Z=X\times Y$, $\chi=\phi\times \psi$, $z_0=(x_0,y_0)$, and $h(x,y)=f(x)-g(y)$ for $(x,y)\in X\times Y$.
Then $$c(n):= h(\chi^n(z_0)) = f(\phi^n(x_0))-g(\psi^n(y_0))= a(n)-b(n),$$ is itself a dynamical sequence and to check that $a(n)=b(n)$ for all $n$ it suffices to check that $c(n)=0$ for all $n$.  
So we will now assume that we are given geometric data $(Z,\chi,z_0,h)$ that generates a dynamical sequence $c(n)$, and we will show how to check that $c(n)=0$ for all $n\ge 0$.

To do this we let $Z_{-1}=Z$ and for each $i\ge 0$ at step $i$ we compute the set $Z_i$, which is defined to be the Zariski closure of the set
\begin{equation} \{z\in Z \colon h(z)=h\circ\chi(z)=\cdots =h\circ \chi^i(z)=0\}.\end{equation}  Since we consider rational maps in general, these maps need not be defined on all of $Z$, which necessitates taking the Zariski closure. We terminate at step $n$ if $Z_n = Z_{n-1}$.

We observe that $Z_{-1}, Z_0, Z_1, \ldots $ are Zariski closed sets and by construction $$Z_{-1}\supseteq Z_0\supseteq Z_1\supseteq \cdots \supseteq Z_n \supseteq \cdots $$ and since the Zariski topology is noetherian (cf. \cite[Example 1.4.7, p. 5]{Hart}), this chain must terminate and so there is some $n_0$ such that $Z_{n_0+1}=Z_{n_0}$.  Hence this procedure terminates. We then claim that $z\in Z_{n_0}\iff h\circ\chi^n(z)=0$ for all $n$.  To see this notice that if $z\not\in Z_0$ then by definition of $Z_{n_0}$ there is some $i\le n_0$ such that $h\circ \chi^i(z)\neq 0$, so one direction is clear.  Now suppose that $z\in Z_{n_0}$.  Then we claim that $h(\chi^n(z))=0$ for all $n$.  By definition this holds for all $n\le n_0$, so if the result does not hold then there is some smallest $m$ such that 
$h(\chi^m(z))\neq 0$ and $m>n_0$.
The equality $Z_{n_0+1}=Z_{n_0}$ implies that if
$$(h(u)=h\circ \chi(u)=\cdots = h\circ \chi^{n_0}(u) =0) \implies h\circ \chi^{n_0+1}(u)=0.$$
But now if we take $u=\chi^{m-n_0-1}(z)$ then $u\in Z_{n_0}$ and so 
$h\circ \chi^{n_0+1}(u) = h(\chi^m(z))=0$, a contradiction.  Thus we get the claim.  Hence after we have determined the value of $n_0$, we can check whether $c(n)=0$ for all $n$ by checking that $c(n)=0$ for $n\le n_0$.
Thus we have shown the following.
\begin{thm}
Let $(Z, \chi, z_0, h)$ be geometric data defining a dynamical sequence $c(n)$. The procedure above terminates and determines whether $c(n) = 0$ for all $n \in \mathbb{N}$. In particular, one can decide whether $a(n) = b(n)$ for all $n$ by checking finitely many terms.
\end{thm}

\begin{remark}
\label{rem:bounds}
Effective bounds on the index $n_0$ at which the chain $\{Z_i\}$ stabilizes can be obtained using intersection theory and generalized Bézout estimates. More specifically, if we embed the variety $Z$ into a projective space, then we can use this embedding to give a notion of degree for subvarieties and maps. Work of Schmid \cite{Sch} shows how one can bound the number of $i$-dimensional components of an intersection of varieties in terms of the degrees of the irreducible components of the two varieties and from this one can give bounds on $n_0$ that depend only on the degrees of our maps. Further details can be found in the Master's thesis of the second-named author \cite{Sun}.  However, such bounds are often far from optimal in concrete applications.
\end{remark}
To illustrate the algorithm in action, we begin with a simple example, which we perform ``by hand.'' Although this manual computation is arguably more complicated than the actual proof, we do this to illustrate how one performs the steps of the algorithm with a concrete example. In subsequent examples, we will use SageMath to perform the algorithm.
\begin{example}
The identity
$$1+2+\cdots + n = n(n+1)/2$$ can be phrased as the equality of two dynamical sequences: if we let $X=\mathbb{A}^2$, $\Phi(x,y)=(x+1,y+x)$, and $f(x,y)=y$, then it is straightforward to see that with each iteration we are increasing the $y$ value by the value of $x$ and then incrementing $x$ by $1$ and so 
$\Phi^n(1,0)=(n+1,1+2+\cdots+ n)$, and hence $$f(\Phi^n(1,0))= \sum_{i=0}^n i.$$ On the other hand, if $Y=\mathbb{A}^1$, $\Psi(t)=t+1$, $g(t)=t(t+1)/2$, then $g(\Psi^n(0))= n(n+1)/2$.
To check the equality of $a(n)=b(n)$ for all $n$, we let
$$Z=\mathbb{A}^3=X\times Y$$ and let $\chi(x,y,z)=(x+1,y+x,z+1)$, and let $h(x,y,z)=y-z(z+1)/2$.
Then to use the algorithm, we let $Z_{-1}=Z=\mathbb{A}^3$.  We next compute $$Z_0=\{(x,y,z)\colon h(x,y,z)=0\}.$$  This is just the variety 
$$Z_0:= \{y=z(z+1)/2 \}.$$ To compute $Z_1$ we find the zero set of $h(x,y,z)$ and $$h\circ \chi(x,y,z)=h(x+1,y+x,z+1)= y+x- (z+1)(z+2)/2.$$  Then 
$$Z_1=\{(x,y,z)\colon y=z(z+1)/2, x = (z+1)(z+2)/2 - z(z+1)/2\}.$$
Next we compute $Z_2$ and we note that $h\circ \chi^2(x,y,z)= h(x+2, y+2x+1, z+2)=y+2x+1 - (z+2)(z+3)/2$.  We check that $Z_1$ is contained in the zero set of $y+2x+1-(z+2)(z+3)/2$ by checking that $y+x-(z+1)(z+2)/2$ vanishes whenever $y=z(z+1)/2, x=(z+1)(z+2)/2 - z(z+1)/2 = (z+1)$, since
$$z(z+1)/2 + 2(z+1)+1 = (z+2)(z+3)/2.$$  It follows that $Z_2=Z_1$ and so to check $1+\cdots +n = n(n+1)/2$, it suffices to check that the identity holds for $n=0,1$, which is indeed the case.
\end{example}
The preceding example falls under the umbrella of the WZ-method, but in general there are identities that can be checked which come from non-holonomic sequences.  We give a few examples---all of the following examples can be proven without computer assistance, but they are still in some cases non-obvious and have been selected to illustrate how the method works to give an automatic verification.  When applicable, we give the corresponding identifier from the OEIS (Online Encyclopedia of Integer Sequences). 

While implementing the algorithm, we monitored termination by generating Gröbner bases step by step, which allowed us to determine when the computation concluded. This is why, in our code, it appears as though we have prior knowledge of when the algorithm terminates. In general, one could instead use a while loop and let the algorithm determine termination automatically.
\subsection{A058635: the sequence $F_{2^n}$}
We let $F_m$ denote the $m$-th Fibonacci number.  We consider the dynamical sequence produced via the geometric data $(\mathbb{A}^1,\Psi,y_0,f)$ where
$\Psi : \mathbb{A}^1 \to \mathbb{A}^1$ is the map $\Psi(t)=t^2$, $y_0 = \rho^2 = (3+\sqrt{5})/2$, and 
\[
f(t) = \frac{1}{\sqrt{5}}\left(t - \frac{1}{t}\right).
\]
Then, for all $n \geq 0$, 
\[
f(\Psi^n(y_0)) = \frac{1}{\sqrt{5}}\left(\rho^{2^{n+1}} - \rho^{-2^{n+1}}\right) = F_{2^{n+1}}.
\]
Observe that if we now consider the geometric data
$(\mathbb{A}^2,\Phi, x_0, g)$, where
$\Phi : \mathbb{A}^2 \to \mathbb{A}^2$ is the morphism
\[
\Phi(x, y) = (xy, y^2 - 2),
\]
$x_0=(1,3)$ and 
$g(x,y)=x$, then we
compute:
\[
g(x_0) = 1, \quad g(\Phi(x_0)) = 3, \quad g(\Phi^2(x_0)) = 21, \ \ldots.
\]
We now verify that $a(n):= g(\Phi^n(x_0))$ equals $F_{2^{n+1}}=f(\Psi^n(y_0))$ for all $n \geq 0$.  To show this, we take $Z=\mathbb{A}^2\times \mathbb{A}^1$ and let $\chi(x,y,z)=(xy,y^2-2,z^2)$, $z_0=(1,3,\rho^2)$ and let $h(x,y,z)=x-(z-1/z)/\sqrt{5}$. 

We now use SageMath. We need to work over the field $\mathbb{Q}(\sqrt{5})$ and we need to invert $z$, so to work with $\mathbb{Q}(\sqrt{5})[x,y,z,z^{-1}]$ we instead work with the polynomial ring $\mathbb{Q}(\sqrt{5})[x,y,z,t]$ and impose the relation $zt-1=0$.  This is the code we used to compute and compare Gr\"obner bases for the ideals generated $h,h\circ \chi, h\circ \chi^2$ and $h,h\circ \chi$:
\begin{verbatim}
# Step 1: Define the field Q(sqrt(5))
P.<s> = PolynomialRing(QQ)
K.<r> = NumberField(s^2 - 5)

# Step 2: Polynomial ring over K in x, y, z, t with lex order
R.<x, y, z, t> = PolynomialRing(K, order='lex')

# Step 3: Define the original discrepancy polynomial h
h0 = x - (z - t)/r  # h(x, y, z), with t = 1/z

# Step 4: Define the chi map (on x, y, z) as substitution
def chi(poly):
    return poly.subs({
        x: x * y,
        y: y^2 - 2,
        z: z^2,
        t: t^2
    })

# Step 5: Compute h, h(chi), h(chi^2)
h1 = chi(h0)
h2 = chi(h1)

# Step 6: Impose the relation t = 1/z via g = zt - 1
g = z * t - 1

# Step 7: Build the two ideals
I1 = R.ideal([h0, h1, g])
I2 = R.ideal([h0, h1, h2, g])

# Step 8: Compute Groebner bases
G1 = I1.groebner_basis()
G2 = I2.groebner_basis()

# Step 9: Compare the two bases
print("GBs equal at step 2?", set(G1) == set(G2))
\end{verbatim}
We get the output:
\begin{verbatim}
GBs the same? True
\end{verbatim}
So it suffices to check that $a(n)=F_{2^{n+1}}$ for $n=0,1$ (which is indeed the case) to obtain that it holds for all $\mathbb{N}$.

\subsection{A000178: Superfactorials, the product of the first $n$ factorials.}
This is the sequence $a(n)$ with initial values 
$$a(0)=1, a(1)= 1, a(2)= 2, a(3)= 12, a(4)= 288, a(5)=34560, \ldots.$$  It's relatively straightforward to make a dynamical sequence producing the superfactorials by taking 
the geometric data $(\mathbb{A}^3,\Phi, (1,1,1), f)$ where $\Phi(x,y,z)=(x+1,xy,xyz)$ and $f(x,y,z)=z$.  
Notice that a straightforward induction shows that $$\Phi^n(1,1,1)=\left(n+1, n!, \prod_{i=0}^n i!\right),$$ and so $$f(\Phi^n(1,1,1)) = \prod_{i=0}^n i!,$$ the $n$-th superfactorial. 
The OEIS page gives a recurrence for the superfactorials due to Somos: if $b(n)$ is given by
$$b(n+3) = (b(n)b(n+2)^3 + b(n+1)^2b(n+2)^2)/b(n+1)^3$$ for all $n$ where $b(-1)=b(0)=b(1)=1$, then $b(n)=a(n)$ for all $n$.  Thus if we let
$p=(1,1,1)=(b(-1),b(0),b(1))$ and let $\Psi:\mathbb{A}^3\to \mathbb{A}^3$ be the map $\Psi(x,y,z)= (y,z, (xz^3+y^2z^2)/y^3)$, then we verify by induction that 
$b(n)=g(\Psi^n(1,1,1))$, where $g(x,y,z)=y$.  We can then run the algorithm to verify that $a(n)=b(n)$ for all $n$, which we do here.

Here we let $Z=\mathbb{A}^3\times \mathbb{A}^3$, let $$\chi(x_1,y_1,z_1, x_2,y_2,z_2) = (x_1+1, x_1y_1, x_1y_1z_1, y_2, z_2, (x_2 z_2^3+y_2^2 z_2^2)/y_2^3)$$ and let $h$ denote the map $z_1-y_2$.  We compute the Gr\"obner basis for the ideal generated by the relations $h=h\circ \chi =h\circ \chi^2=0$ and compare it to the Gr\"obner basis for the ideal generated by the relations $h=\cdots = h\circ \chi^3=0$.  We must clear denominators at intermediate steps to do this computation in SageMath, but this does not affect the validity of the termination argument for the algorithm, since our orbits avoid the indeterminacy loci of the maps.  
We use the following SageMath commands:
\begin{verbatim}
# Polynomial ring setup
A = PolynomialRing(QQ, ['x1','y1','z1','x2','y2','z2'], order='degrevlex')
x1, y1, z1, x2, y2, z2 = A.gens()

# Discrepancy function
h = z1 - y2

# Corrected chi map
def chi(v):
    x1, y1, z1, x2, y2, z2 = v
    return (
        x1 + 1,
        x1 * y1,
        x1 * y1 * z1,
        y2,
        z2,
        (x2 * z2^3 + y2^2 * z2^2) / y2^3
    )

# Cleared h(chi^n)
def cleared_h(v):
    expr = h.subs(dict(zip((x1, y1, z1, x2, y2, z2), v)))
    return expr.numerator()

# Generate h_0 through h_3
v = (x1, y1, z1, x2, y2, z2)
h_polys = []
for _ in range(4):  # up to h(chi^3)
    h_i = cleared_h(v)
    h_polys.append(h_i)
    v = chi(v)

# Compute G2 and G3
I2 = A.ideal(h_polys[:3])
I3 = A.ideal(h_polys[:4])

G2 = I2.groebner_basis()
G3 = I3.groebner_basis()

# Compare
print("GBs the same? ", set(G2) == set(G3))
\end{verbatim}
This confirms that it suffices to verify the identity for $n \le 2$, which is indeed the case.
\subsection{A006769: Elliptic divisibility sequence associated with $y^2 + y = x^3 - x$ and multiples of the point $(0,0)$}
This is the sequence $a(n)$ whose initial values (starting with index $n=0$) are given by $$0, 1, 1, -1, 1, 2, -1, -3, -5, 7, -4, -23, 29, 59, 129, -314, \ldots.$$
The OEIS page gives two recurrences for this sequence:
$$a(n) = (a(n-1)\cdot  a(n-3) + a(n-2)^2) / a(n-4)$$ for all $n\ge 5$ and
$$a(n) = (-a(n-1) \cdot a(n-4) - a(n-2) \cdot a(n-3)) / a(n-5)$$ for all $n\ge 6$.  We note that by Remark \ref{rem:Mazur} this sequence never takes the value zero.
We let $X=\mathbb{A}^4$ and to model the first recurrence, we let 
$\phi : X\to X$ be the map
$\phi(x,y,z,w) = (y,z,w, (wy+z^2)/x)$, so by construction $$\phi(a(i), a(i+1),a(i+2),a(i+3)) = (a(i+1), a(i+2), a(i+3), a(i+4))$$ for all $i\ge 1$.  In particular, if $f(x,y,z,w)=x$ then
$(X,\phi, (1,1,-1,1), f)$ generates the sequence $a(n+1)$ for all $n\ge 0$.  
Similarly, if $Y=\mathbb{A}^5$, we can model the second recurrence by taking $\psi(x,y,z,w,t) = (y,z,w,t, -(ty+zw)/x)$ and if we take $g(x,y,z,w,t)=x$ then
we see in a similar manner that $(Y,\psi, (1,1,-1,1,2),g)$ also generates the sequence $a(n+1)$. We can verify that these recurrences both generate $a(n+1)$ by applying the algorithm.  Here we take $Z=\mathbb{A}^9$, 
\begin{align*}
&~ \chi(x_1,y_1,z_1,w_1, x_2,y_2,z_2,w_2,t_2)\\
&=( y_1,z_1,w_1, (w_1y_1+z_1^2)/x_1, y_2, z_2,w_2,t_2, -(t_2y_2+z_2w_2)/x_2),
\end{align*}
 and let 
$h=x_1-x_2$, then if we compute the Gr\"obner basis for the relations $h=h\circ \chi= \cdots =  h\circ \chi^5=0$ and the Gr\"obner basis for the relations
$h=h\circ \chi= \cdots =  h\circ \chi^6=0$ and compare them. We find that the resulting Gröbner bases are equal.
In this case, due to the introduction of denominators, we have to clear denominators to compute. However, since our orbits avoid the indeterminacy loci, this is unimportant for verification.

We use the following SageMath commands:
\begin{verbatim}
# Set up polynomial ring over Q
A = PolynomialRing(QQ, ['x1','y1','z1','w1','x2','y2','z2','w2','t2'])
x1, y1, z1, w1, x2, y2, z2, w2, t2 = A.gens()

# Define h = x1 - x2
h = x1 - x2

# Define the map chi
def chi(v):
    x1, y1, z1, w1, x2, y2, z2, w2, t2 = v
    return (
        y1,
        z1,
        w1,
        (w1 * y1 + z1^2) / x1,
        y2,
        z2,
        w2,
        t2,
        -(t2 * y2 + z2 * w2) / x2
    )

# Generate cleared iterates of h (chi^i)
def cleared_h_iterates(max_iter):
    v = (x1, y1, z1, w1, x2, y2, z2, w2, t2)
    iterates = [v]
    h_polys = []
    for _ in range(max_iter + 1):
        expr = h.subs(dict(zip(v, iterates[-1])))
        cleared = expr.denominator() * expr  # clear denominators
        h_polys.append(cleared)
        iterates.append(chi(iterates[-1]))
    return h_polys

# Compute Groebner bases
G5 = A.ideal(cleared_h_iterates(5)).groebner_basis()
G6 = A.ideal(cleared_h_iterates(6)).groebner_basis()

# Compare
print("GBs the same?", set(G5) == set(G6))
\end{verbatim}
This produces the output:
\begin{verbatim}
GBs the same? True
\end{verbatim}
Thus the identity can be verified by checking that it holds for $n\le 5$, which is the case.
\subsection{A006720: Somos-$4$ sequence: $b(0)=b(1)=b(2)=b(3)=1$, $b(n) = (b(n-1) b(n-3) + b(n-2)^2) / b(n-4)$ for $n\ge 4$.}
This sequence satisfies the same recurrence as the preceding elliptic divisibility sequence, but has different initial values. Its first few terms are given by
$$1, 1, 1, 1, 2, 3, 7, 23, 59, 314, 1529, 8209, 83313, 620297, \ldots.$$
The OEIS page notes (using slightly different indexing) that $$b(n)=a(2n-3)(-1)^n\qquad {\rm for}~ n\ge 2.$$ This identity is non-obvious but can be easily verified using the algorithm we describe. To do this, we note that the values of $b(n)$ are all positive, and we can encode the recurrence by noting that 
$\phi(b(i),b(i+1), b(i+2), b(i+3))=(b(i+1),b(i+2), b(i+3), b(i+4))$ when $\phi:\mathbb{A}^4\to \mathbb{A}^4$ is the map
$\phi(x,y,z,w)=(y,z,w, (wy+z^2)/x)$.  Hence if we let $p_0=(b(2),b(3),b(4),b(5))=(1,1,2,3)$ then $f\circ \phi^n(p_0)=b(n+2)$ for all $n\ge 0$, where $f(x,y,z,w)=x$.  We wish to show that 
$b(n+2) = (-1)^n a(2n+1)$ for all $n\ge 0$.  The sequence $a(n)$ satisfies the same recurrence as $b(n)$, and we know from the preceding subsection that 
$a(n+1)=f(\phi^n(1,1,-1,1))$ for all $n\ge 0$, where we are using the same maps $\phi$ and $f$ as the ones that generate the sequence $b(n+2)$.  In particular, if we let $\psi=\phi^2$ then $f(\psi^n(1,1,-1,1))= f(\phi^{2n}(1,1,-1,1))=a(2n+1)$.  We compute
$$\psi(x,y,z,w)=\phi(y,z,w,(wy+z^2)/x) = (z,w, (wy+z^2)/x, (zx^{-1}(wy+z^2) +w^2)y^{-1}).$$ 
$$f(\psi^n(1,1,-1,1))=f(\phi^{2n}(1,1,-1,1)) = a(2n+1)$$ for all $n\ge 0$.  
To keep track of the alternating sign, we can add an extra coordinate whose value alternates between $1$ and $-1$ at each step:
let $\chi: \mathbb{A}^5\to \mathbb{A}^5$ be given by $\chi(x,y,z,w,t)= (\psi(x,y,z,w), -t)$ and let $g(x,y,z,w,t)=xt$. Then
 $$g(\chi^n(1,1,-1,1,1))= a(2n+1)(-1)^n.$$  We can now verify that $g(\chi^n(1,1,-1,1,1))=f(\phi^n(1,1,2,3))$ for all $n\ge 0$ via our algorithm.
 Here we let $Z=\mathbb{A}^9$ and let $\Theta: Z\to Z$ be the map
 \begin{align*}
 &~\Theta(x_1,y_1,z_1,w_1,x_2,y_2,z_2,w_2,t)\\
 &= (y_1,z_1,w_1, (w_1y_1+z_1^2)/x_1,  z_2,w_2, (w_2y_2+z_2^2)/x_2,  (z_2x_2^{-1}(w_2y_2+z_2^2) +w_2^2)y_2^{-1},-t),
 \end{align*}
  and let
 $h$ be the rational map $x_1-x_2t$. 
 We use the following SageMath commands to show that the Gr\"obner bases for the ideals (with cleared denominators) generated by $h\circ \Theta^i$ for $i\le 6$ and for $h\circ \Theta^i$ for $i\le 7$ are the same:
\begin{verbatim}
# Polynomial ring over Q with 9 variables
A = PolynomialRing(QQ, ['x1','y1','z1','w1','x2','y2','z2','w2','t'], 
order='degrevlex')
x1, y1, z1, w1, x2, y2, z2, w2, t = A.gens()

# Discrepancy function
h = x1 - x2 * t

# Theta map
def Theta(v):
    x1, y1, z1, w1, x2, y2, z2, w2, t = v
    return (
        y1,
        z1,
        w1,
        (w1 * y1 + z1^2) / x1,
        z2,
        w2,
        (w2 * y2 + z2^2) / x2,
        ((z2 / x2) * (w2 * y2 + z2^2) + w2^2) / y2,
        -t
    )

# Cleared denominators for h(Theta^n)
def cleared_h(v):
    expr = h.subs(dict(zip((x1, y1, z1, w1, x2, y2, z2, w2, t), v)))
    return expr.numerator()

# Generate up to h_7
v = (x1, y1, z1, w1, x2, y2, z2, w2, t)
h_polys = []
for _ in range(8):  # go to h_7
    h_i = cleared_h(v)
    h_polys.append(h_i)
    v = Theta(v)

# Construct ideals
I6 = A.ideal(h_polys[:7])
I7 = A.ideal(h_polys[:8])

# Compute Groebner bases
G6 = I6.groebner_basis()
G7 = I7.groebner_basis()

# Compare
print("GBs the same?", set(G6) == set(G7))
\end{verbatim}
Interestingly, in all other cases we considered, the program terminated almost instantly; in this case, however, it took approximately three seconds to run and confirmed that the Gr\"obner bases are equal. Hence, to verify that $a(n) = b(n)$ for all $n$, it suffices to check that $a(n) = b(n)$ for $n \leq 6$, which is straightforward.

\subsection{The identity $(1+x)(1+x^2)\cdots (1+x^{2^{n-1}}) = (x^{2^n}-1)/(x-1)$}
While this identity is not difficult to prove, we can give an easy ``automatic'' proof. For $\lambda\in \mathbb{C}$, with $\lambda\neq 1$, we can define two dynamical sequences
$$a(n)=(1+\lambda)(1+\lambda^2)\cdots (1+\lambda^{2^{n-1}})\qquad {\rm and} \qquad b(n)=(\lambda^{2^n}-1)/(\lambda-1).$$  To generate $b(n)$, we just iterate the squaring map on $\mathbb{A}^1$, so we let $\Phi(x)=x^2$ and let $f(x)=(x-1)/(\lambda-1)$.  Then since $\Phi^n(\lambda)=\lambda^{2^n}$, $f(\Phi^n(\lambda))= (\lambda^{2^n}-1)/(\lambda-1)=b(n)$.  On the other hand, we can realize $a(n)$ as a dynamical sequence by taking 
the map $\Psi:\mathbb{A}^2\to \mathbb{A}^2$ given by $\Psi(x,y)=(x^2, (x+1)y)$.  Then by induction 
$$\Psi^n(\lambda,1)=\left(\lambda^{2^n}, \prod_{i=0}^{n-1} (1+\lambda^{2^i})\right),$$ with the convention that an empty product equals $1$.  Then if $g(x,y)=y$,
$$a(n)=g(\Psi^n(\lambda,1)).$$ We now employ the algorithm to check that $a(n)=b(n)$ for all $n$.  

We work over the field $\mathbb{Q}(t)$, and we let $Z=\mathbb{A}^2\times \mathbb{A}^1$ and let $\chi(x,y,z)=(x^2, (x+1)y, z^2)$ and let $h(x,y,z)=y - (z-1)/(t -1)$.  
The following code in SageMath verifies that the Gr\"obner bases for the ideal of $\mathbb{Q}(t)[x,y,z]$ generated by $h, h\circ \chi$, and the ideal generated by  $h, h\circ \chi, h\circ \chi^2$ are the same, and so to verify the identity for all $n$, it suffices to check that it holds for $n=0,1$:
\begin{verbatim}
# Define base field
R.<t> = QQ[]

# Polynomial ring in three variables over Q(t)
A.<x, y, z> = PolynomialRing(Frac(R))

# Define h
h = y - (z - 1)/(t - 1)

# Define chi as a substitution dictionary
def chi(poly):
    return poly(x^2, (x + 1)*y, z^2)

# Compose h with chi
h1 = chi(h)
h2 = chi(h1)

# Generate the ideals
I1 = ideal([h, h1])
I2 = ideal([h, h1, h2])

# Compute Groebner bases
G1 = I1.groebner_basis()
G2 = I2.groebner_basis()

# Check if the second ideal gives anything new
print("GBs the same?", set(G1) == set(G2))
\end{verbatim}
Running this code gives the output:
\begin{verbatim}
GBs the same? True
\end{verbatim}
Thus we recover the classical identity
$(1+t)(1+t^2)\cdots (1+t^{2^{n-1}}) = (t^{2^n}-1)/(t-1)$ for $t\neq 1$.

\section{Open questions}
\label{sec:questions}

In this section, we collect several open problems concerning dynamical sequences.

We begin by revisiting the closure properties given in Theorem~\ref{thm:main}. Notably absent from this list is closure under convolution. It remains open whether the class of dynamical sequences is closed under this operation.

\begin{quest} \label{quest:1}
Let $a(n)$ and $b(n)$ be dynamical sequences over a field $K$. Is the sequence
\[
c(n) := \sum_{i=0}^n a(i) b(n-i)
\]
necessarily a dynamical sequence over $K$?
\end{quest}

We suspect the answer is \emph{no}. However, if the answer were affirmative, the generating functions of dynamical sequences would form a subring of the ring of formal power series.

\medskip

The dynamical Mordell–Lang conjecture (see \cite{DML-book}) asserts that if $X$ is a variety defined over a field $K$ of characteristic zero, $\phi:X\dasharrow X$ is a rational map, $V \subseteq X$ is a subvariety, and $x_0 \in X$ is a point whose orbit avoids the indeterminacy locus of $\phi$, then the set of $n$ for which $\phi^n(x_0)\in V$ is a finite union of infinite arithmetic progressions together with a finite set. 

If we take $V$ to be the zero set of a rational map $f$, this suggests that the zero set of a dynamical sequence over a field of characteristic zero should have a similar form, generalizing the classical Skolem–Mahler–Lech theorem (see \cite[Ch.~2]{Everest} or \cite{Lech}).

\begin{quest}
Let $a(n)$ be a dynamical sequence over a field of characteristic zero. Is the set of $n$ for which $a(n) = 0$ a finite (possibly empty) union of infinite arithmetic progressions along with a finite set?
\end{quest}

We note that this fails in characteristic $p > 0$, even for linear recurrences. For example, as shown by Lech \cite{Lech}, let $p$ be a prime, $K = \overline{\mathbb{F}_p(t)}$, and define $\phi(x,y) = (tx, (1 + t)y)$ and $f(x,y) = y - x - 1$. Then $f(\phi^n(1,1)) = 0$ if and only if $n \in \{1, p, p^2, p^3, \ldots\}$.

\medskip

Skolem's problem asks whether there is an algorithm to decide whether a given linear recurrence sequence takes the value zero at some term. We note that this problem has received considerable attention recently (see for example \cite{Sko}). We pose the analogous question for dynamical sequences.

\begin{quest}
Let $K$ be a number field and let $a(n)$ be a $K$-valued dynamical sequence. Is it decidable whether $a(n) = 0$ for some $n$?
\end{quest}

By closure under partial products (Theorem~\ref{thm:main}(4)), this is equivalent to deciding whether $a(n) = 0$ for all sufficiently large $n$.

\medskip

In the case of integer sequences such as certain elliptic divisibility sequences and Somos sequences, it is often not obvious from the recurrence that all terms are integers. This motivates the following question.

\begin{quest}
Let $a(n)$ be a $\mathbb{Q}$-valued dynamical sequence. Is there an algorithm to determine whether $a(n)$ is an integer for all $n$?
\end{quest}

Silverman (see \cite[Theorem 2.2]{Silverman}) has shown that in the case of a rational map $\phi:\mathbb{P}^1\to \mathbb{P}^1$, if $\phi^n(a)\in \mathbb{Z}$ for all $n$ then $\phi^2(x)$ is a polynomial and hence in this case one can check that the orbit is integer-valued by computing the map $\phi^2$, but in general the question remains open.

\medskip

In the case of integer-valued holonomic sequences, the Pólya–Carlson \cite{PC} theorem implies that if the generating function is not rational, then it has the unit circle as a natural boundary. This leads us to ask whether a dynamical analogue of rationality holds under suitable growth constraints.

\begin{quest}
Let $a(n)$ be an integer-valued dynamical sequence and suppose that $|a(n)| = e^{o(n)}$. Must $a(n)$ satisfy a linear recurrence?
\end{quest}

\medskip

Pisot's conjecture (proved by Zannier; see \cite{Zannier}) states that if $b(n)$ is a sequence of nonnegative integers and $b(n)^d$ satisfies a linear recurrence for some $d \ge 1$, then $b(n)$ itself must satisfy a linear recurrence. We ask whether this extends to dynamical sequences.

\begin{quest}
Let $b(n)$ be a sequence of nonnegative integers and suppose that $b(n)^d$ is a dynamical sequence for some $d \ge 1$. Must $b(n)$ also be a dynamical sequence?
\end{quest}

\medskip

We also ask whether one can give a coarse classification of the asymptotic behaviour of heights of $\overline{\mathbb{Q}}$-valued dynamical sequences. See \cite{BHS, BNZ} for related questions.

\begin{quest}
Can one classify all possible height gaps for $\overline{\mathbb{Q}}$-valued dynamical sequences?
\end{quest}

Even for $D$-finite sequences, this remains challenging (see \cite[Question 4.3]{BNZ}), but algebraic dynamical systems may yield more exotic growth patterns due to the underlying geometry.

\medskip

Finally, we note that many interesting binary sequences arise naturally as sign patterns of integer-valued dynamical sequences. Silverman and Stephens \cite{SilStur} have shown that many Sturmian sequences can be obtained in this manner from Elliptic divisibility sequences.

\begin{quest}
Let $a(n)$ be an integer-valued dynamical sequence and define $\epsilon(n) = \operatorname{sgn}(a(n)) \in \{\pm 1\}$. Which binary sequences can arise in this way? Can one characterize or give natural classes of binary sequences that cannot be realized as sign patterns of dynamical sequences?
\end{quest}
In the case of linear recurrences, this is the question of understanding the positivity set (see \cite{BG}).

\end{document}